\documentclass[a4paper,twocolumn,11pt,accepted=2024-05-03]{quantumarticle}
\pdfoutput=1
\usepackage[utf8]{inputenc}
\usepackage[english]{babel}
\usepackage[T1]{fontenc}
\usepackage{amsmath}
\usepackage{hyperref}
\usepackage{amsfonts}
\usepackage{amssymb}
\usepackage{amsthm}
\usepackage{graphicx}
\usepackage{fullpage}
\usepackage[capitalize]{cleveref}
\usepackage{tikz-network}
\usepackage{tikz-cd}
\usepackage{xcolor}
\graphicspath{{figures/}}

\RequirePackage{subcaption}

\newtheorem{theorem}{Theorem}

\newtheorem{lemma}[theorem]{Lemma}
\newtheorem{corollary}[theorem]{Corollary}

\Crefname{lemma}{Lemma}{Lemma}
\Crefname{theorem}{Theorem}{Theorem}
\Crefname{conjecture}{Conjecture}{Conjecture}
\Crefname{corollary}{Corollary}{Corollary}
\Crefname{definition}{Definition}{Definition}

\newcommand{\im}{\operatorname{im}}
\newcommand{\id}{\operatorname{id}}

\newcommand{\Field}{\mathbb{F}}
\newcommand{\Integer}{\mathbb{Z}}
\newcommand{\Complex}{\mathbb{C}}
\newcommand{\Real}{\mathbb{R}}
\newcommand{\Aut}{\mathrm{Aut}}

\newcommand{\Disc}{\mathbb{D}}
\newcommand{\Homology}{\mathcal{H}}

\newcommand{\vspac}[1]{\Field_2^{#1}}

\newcommand{\csscode}{C}
\newcommand{\underlying}{C/\sim_{ZX}}
\newcommand{\zxdualities}{\mathcal{D}_{ZX}}

\newcommand\blfootnote[1]{%
	\begingroup
	\renewcommand\thefootnote{}\footnote{#1}%
	\addtocounter{footnote}{-1}%
	\endgroup
}

\begin{document}
	\title{Fold-Transversal Clifford Gates for Quantum Codes}
	
	\author{Nikolas P. Breuckmann}
	\affiliation{Department of Computer Science, University College London, WC1E 6BT London, United Kingdom}
	\orcid{0000-0002-7211-5515}
	\email{niko.breuckmann@bristol.ac.uk}
	\homepage{http://nikobreu.website}

    \author{Simon Burton}
	\affiliation{Institute of Physics, Jagiellonian University, \L{}ojasiewicza 11, 30-348 Krak\'ow, Poland}
	\orcid{0000-0002-8932-3492}
	\email{simon@arrowtheory.com}
	\homepage{http://arrowtheory.com}
	
	\maketitle
	\blfootnote{The authors contribute equally to this work.}

	\begin{abstract}
		We generalize the concept of folding from surface codes to CSS codes by considering certain dualities within them.
		In particular, this gives a general method to implement logical operations in suitable LDPC quantum codes using transversal gates and qubit permutations only.
		
		To demonstrate our approach, we specifically consider a $[[30,8,3]]$ hyperbolic quantum code called Bring's code.
		Further, we show that by restricting the logical subspace of Bring's code to four qubits, we can obtain the \emph{full} Clifford group on that subspace.
	\end{abstract}

	\section{Introduction}
	We show how symmetries of CSS quantum codes can be utilized to implement encoded Clifford gates, using only transversal gates and qubit permutations.
	Our scheme is not only trivially fault-tolerant, but it also incurs no time-overhead and does not require any additional ancilla qubits.
	It can be understood as a generalization of the concept of folding surface codes due to Moussa~\cite{moussa2016transversal}, see \Cref{fig:ZX-dualities}.
	
	Folding of topological quantum codes is part of the folklore relating surface codes and color codes:
	folded surface codes give color codes, or alternatively, unfolded color codes give surface codes~\cite{Kubica_2015}.

	\begin{figure*}[t]
	  \centering
	   \begin{subfigure}[b]{0.3\textwidth}
	      \includegraphics[width=\textwidth]{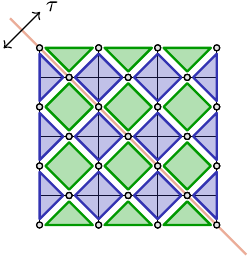}
	      \caption{
	        Moussa's original construction on
	        a surface code.
	        The $ZX$-duality $\tau$ is reflecting along a diagonal.
	        }
	      \label{fig:fold-surface}
	   \end{subfigure}
	    ~
	   \begin{subfigure}[b]{0.3\textwidth}
	      \includegraphics[width=\textwidth]{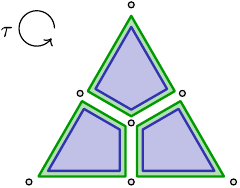}
	      \caption{
	        The Steane code is self-dual, and so the identity
	        permutation serves as a $ZX$-duality.
	        }
	      \label{fig:fold-steane}
	   \end{subfigure}
	    ~
	   \begin{subfigure}[b]{0.3\textwidth}
	      \includegraphics[width=\textwidth]{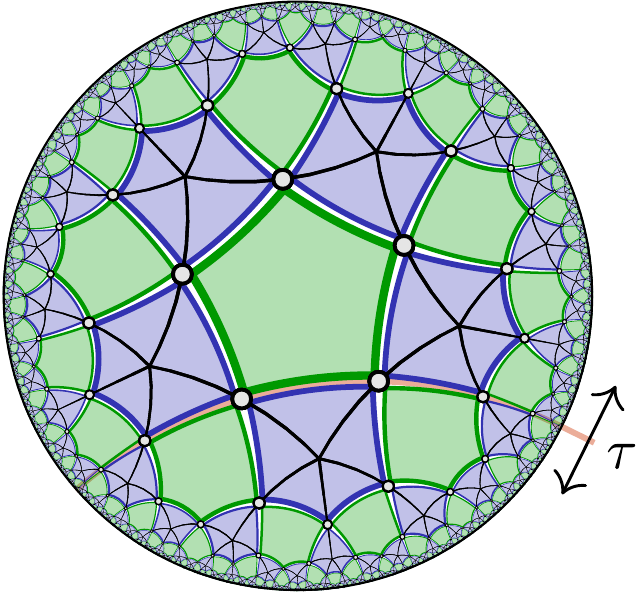}
	      \caption{
	        On a $\{5,5\}$-tiled hyperbolic surface, a $ZX$-duality is shown reflecting along a geodesic.
	        }
	      \label{fig:fold-hyperbolic-55}
	   \end{subfigure}
		\caption{How $ZX$-dualities define folds of quantum codes. 
	        Qubits are shown as grey circles.
	        The green/blue faces correspond to the $Z$- and $X$-checks. 
	        The $ZX$-duality $\tau$ maps between the two types of checks.
	        The codes (i) and (iii) are built from the homology of a tiling (cellulation) shown in black.
	        }
		\label{fig:ZX-dualities}
	\end{figure*}

	We generalize the concept of folding to arbitrary CSS codes by introducing \emph{fold-transversal gates}.
	While our approach works for any sufficiently symmetric CSS code, we believe it is particularly suited to implement logical gates in low-density parity-check (LDPC) quantum codes.
	These codes  have recently attracted a lot of attention, see~\cite{PRXQuantum.2.040101} for a recent review.
	This is partially due to the fact that constant overhead fault-tolerant quantum computation can be realized with these codes~\cite{gottesman2014overhead,fawzi2018constant}.
	Recent developments have shown that it is possible to construct high-performance LDPC quantum codes from classical codes~\cite{hastings2021fiber,panteleev2021quantum,breuckmann2021balanced}.
	In particular, \cite{breuckmann2021balanced} gave a construction that was conjectured to have optimal asymptotic parameter scaling, which was subsequently proven in~\cite{panteleev2021asymptotically}.
	
	While the construction of LDPC quantum codes has made big advances, it is still unclear how to best manipulate the encoded information.
	The constant overhead in~\cite{gottesman2014overhead} was achieved by implementing gates via ancillary states.
	However, in order to keep the qubit overhead constant in their proof, the implementation of gates had to be serialized, making the scheme less practical.
	
	For specific families of LDPC quantum codes, a few techniques to fault-tolerantly implement logical gates have been introduced.
	For hyperbolic surface codes, Breuckmann et~al. considered code deformations in order to perform $\operatorname{CNOT}$-gates~\cite{hyperbolic_quantum_storage}.
	Further, Lavasani--Barkeshli discussed the implementation of other Clifford gates on hyperbolic surface codes by using ancillary states~\cite{lavasani2018low}.
	Krishna--Poulin~\cite{krishna2021fault} considered generalizations of code deformation techniques of the surface code to hypergraph product codes in order to implement Clifford gates.
	On the other hand, Burton--Browne~\cite{burton2021limitations} showed that it is not possible to obtain transversal logical gates outside of the Clifford group using certain hypergraph product codes.
	
	More recently, Cohen et al. generalized lattice surgery between two logical operators of an arbitrary LDPC quantum code~\cite{cohen2021low}.
	The surgery is facilitated by an ancillary surface code that grows with the code distance.
	Hence, for codes with large distance, it is necessary to serialize the quantum circuit to avoid a large overhead, as in~\cite{gottesman2014overhead}.
	
	In comparison, the logical gates constructed in our work do not incur any qubit overhead, as they are realized within the quantum code directly, avoiding the need for ancilla qubits.
	We also avoid any time overhead, as logical gates are implemented by a transversal circuit of single- and two-qubit gates, which can be applied in parallel.
	This also means that our implementation is compatible with any decoding scheme, as it can be executed between rounds of error correction.
	
	We explicitly construct fault transversal gates for a certain $[[30,8,3]]$ hyperbolic surface code that we call Bring's code.
	We show that we can generate a subgroup of the Clifford group~$\mathcal{C}_8$ using transversal circuits.
	By adding a single logical entangling gate, which can be obtained using other techniques, we generate the full Clifford group~$\mathcal{C}_8$.
	Alternatively, when we restrict the logical subspace to four qubits, we can generate all of~$\mathcal{C}_4$, using transversal circuits only.

	In this work we do not consider any constraints on the locality of qubit interactions.
	In particular, we make use of qubit permutations, which in principle, can map any qubit to any other qubit.
	Thus, in order to take full advantage of our construction, a potential hardware architecture ideally supports non-local interactions between qubits.
	Fortunately, there have been several advances in experimental setups that can realize such unconstrained connectivity.
	Linear optical quantum computation, for example, is naturally flexible in terms of connectivity \cite{rudolph2017optimistic,bombin2021interleaving,bartolucci2021fusion}.
	The Moussa folding construction is used by the fusion-based photonic scheme proposed in~\cite{Kim2021}.
	In a similar vain, modular architectures, in which modules are interlinked by a photonic interface, allow for non-local connectivity~\cite{monroe2013scaling,nigmatullin2016minimally,nickerson2014freely}.
	Other approaches to quantum computation, such as mobile qubits~\cite{bluvstein2021quantum} or qubits coupled to a common cavity mode~\cite{PhysRevLett.75.3788,PhysRevA.94.053830}, even allow for fast any-to-any connectivity between qubits.
	In particular, \cite{ramette2021any} proposes an architecture based on Rydberg atoms, or alternatively, ion-traps, in which it is possible to have fast, long-range interactions between hundreds of physical qubits.
	
	Bring's code is small enough that this scheme could be implemented on NISQ hardware.
	Together with the fold-transversal gates it provides a simple set-up in which fault-tolerant Clifford gates could be demonstrated.
	
In \Cref{sec:folding} we introduce the theory of \emph{fold-transversal} gates.
When discussing stabilizer codes, we largely employ the language of homology 
which is introduced in \Cref{sec:css}.
In \Cref{sec:ZX-dualities}
we give the definition of a \emph{$ZX$-duality} of a homological code, and 
show how the symmetries of the code interact with the $ZX$-dualities.
We then apply this to a class of highly symmetric two-dimensional hyperbolic surface codes
in \Cref{sec:surface-codes}. These codes all have $ZX$-dualities, which are
described in \Cref{thm:hyperbolic}.
Armed with this theory of $ZX$-dualities, in \Cref{sec:fold_transversal_gates}
we characterize when these dualities give rise to \emph{fold-transversal} Clifford gates on the 
logical subspace of the code.
In \Cref{sec:brings_code} we discuss the example of
\emph{Bring's code} and work out in detail the fold-transversal gates
and their action on the logical subspace.
	
	\section{Folding along a $ZX$-duality}\label{sec:folding}
	
	\subsection{CSS quantum codes}\label{sec:css}
	A \emph{stabilizer code} is defined by a subgroup~$\mathcal{S}$ of the Pauli group~$\mathcal{P}_n$ which is abelian and does not contain~$-I$.
	We call the elements of some distinguished set of generators the \emph{(stabilizer) checks} of the stabilizer code.
	A stabilizer code is a \emph{Calderbank-Shor-Steane (CSS) code} if there exists a generating set of~$\mathcal{S}$ such that each generator acts as either Pauli-$X$ or Pauli-$Z$ on all qubits in its support.
	Hence, we can define a CSS code $\csscode$ in terms of two binary matrices $H_X\in \mathbb{F}_2^{r_X\times n}$ and $H_Z\in \mathbb{F}_2^{r_Z\times n}$ where each row corresponds to the support indication vector of a Pauli-$X$ or Pauli-$Z$ generator, respectively.
Note we do not require~$H_X$ or~$H_Z$ to have full rank.
	As $X$-type and $Z$-type Pauli operators commute if and only if the overlap of their supports contains an even number of qubits we must have
	\begin{align}\label{eqn:CSS_commutation}
		H_X H_Z^\top = 0.
	\end{align}
	
	In the language of homology, a CSS code $\csscode$ is equivalent to a chain complex 
 over the field $\Field_2=\Integer/2$:
$$
    \csscode = \{ \vspac{r_Z} \xrightarrow{H_Z^\top} \vspac{n} \xrightarrow{H_X} \vspac{r_X} \}.
$$
We can think of $H_Z^\top$ and $H_X$ as \emph{boundary operators}~$\partial_2$ and~$\partial_1$, respectively:
\begin{center}
\begin{tikzcd}
C_2  \arrow[r, "\partial_2"] & C_1  \arrow[r, "\partial_1"] & C_0
\end{tikzcd}
\end{center}
where $C_2 = \mathbb{F}_2^{r_Z}$, $C_1 = \mathbb{F}_2^n$ and $C_0 = \mathbb{F}_2^{r_X}$.
The elements of~$C_i$ are called \emph{$i$-chains}.
Note that the space of $i$-chains comes with a natural basis due to the construction
of $C_i$ as a free vector space on a basis.
The fact that ``the boundary of a boundary is zero'',
i.e.~$\partial_1\, \partial_2 = 0$, is guaranteed by
\Cref{eqn:CSS_commutation}.
The group generated by the $Z$-type stabilizer operators
$\mathcal{S}_Z$ corresponds to boundaries $B_1 = \im
\partial_2$ and the $Z$-type Pauli-operators normalizing
the $X$-type stabilizer operators $N(\mathcal{S}_X)_Z$
corresponds to cycles $Z_1 = \ker \partial_1$.
Similarly, $\mathcal{S}_X$ and $N(\mathcal{S}_Z)_X$ correspond
to coboundaries $B^1 = \im \delta_0$ and cocycles $Z^1
= \ker \delta_1$, respectively, where $\delta_0 = \partial_1^\top = H_X^\top$
and $\delta_1 = \partial_2^\top=H_Z$.
The $Z$-type logical operators of the code are found 
as cycles modulo boundaries, 
which is the (first) \emph{homology} of the code,
$\Homology_1 := Z_1 / B_1.$
Dually, the $X$-type logical operators 
are cocycles modulo coboundaries.
This is the (first) \emph{cohomology} of the code,
$\Homology^1 := Z^1 / B^1.$
A (co)cycle that is not trivial modulo a (co)boundary is called
an \emph{essential (co)cycle}.

	\subsection{$ZX$-Dualities}\label{sec:ZX-dualities}

%

Given a CSS Code $\csscode=(H_Z,H_X)$,
the \emph{dual} CSS code is
$$
    \csscode^\top = \{ \vspac{r_X} \xrightarrow{H_X^\top} \vspac{n} \xrightarrow{H_Z} \vspac{r_Z} \}.
$$
A \emph{self-dual} CSS code $\csscode$ has $\csscode=\csscode^\top.$

An \emph{isomorphism} of CSS codes $\sigma: \csscode\to \csscode'$ is a
triple of permutation matrices $(\sigma_Z,\sigma_n,\sigma_X)$ where
\begin{align*}
\sigma_Z&:\vspac{r_Z}\to \vspac{r'_Z},\\
\sigma_n&:\vspac{n}\to \vspac{n'},\\
\sigma_X&:\vspac{r_X}\to \vspac{r'_X},
\end{align*}
such that the following diagram commutes:
\[
\begin{tikzcd}
\vspac{r_Z} \arrow{d}{\sigma_Z} \arrow{r}{H^\top_Z}
    & \vspac{n} \arrow{d}{\sigma_n} \arrow{r}{H_X}
    & \vspac{r_X} \arrow{d}{\sigma_X} \\
\vspac{r'_Z} \arrow{r}{H'^\top_Z} & \vspac{n'} \arrow{r}{H'_X} & \vspac{r'_X}
\end{tikzcd}
\]
The \emph{automorphism group} of a CSS code $\csscode$ is
the group of self-isomorphisms $\Aut(\csscode) = \{\sigma:\csscode\to \csscode\}.$

We are also interested in weaker notions of isomorphism between
codes that allow for swapping the $X$- and $Z$-checks.
Given a CSS code $\csscode=(H_Z, H_X)$ we define the 
\emph{underlying classical} code $\underlying$ 
to be given by the parity 
check matrix $H$ as:
$$
\underlying := \Bigl\{
\vspac{n} \xrightarrow{H := \left(\begin{smallmatrix}H_X\\ H_Z \end{smallmatrix}\right)} \vspac{r}
\Bigr\}.
$$
where $r:=r_X + r_Z.$
This classical code remembers all the checks
of the quantum code, but it has forgotten whether each check is
an $X$- or $Z$-type check.

An \emph{automorphism} of a classical code with parity check matrix
$H:\vspac{n}\to\vspac{r}$ is a pair of permutation matrices
$(\sigma_n,\sigma_r)$ where 
\begin{align*}
\sigma_n&:\vspac{n}\to \vspac{n},\\
\sigma_r&:\vspac{r}\to \vspac{r},
\end{align*}
such that the following diagram commutes:
\[
\begin{tikzcd}
\vspac{n} \arrow{d}{\sigma_n} \arrow{r}{H} & \vspac{r} \arrow{d}{\sigma_r} \\
\vspac{n} \arrow{r}{H} & \vspac{r}
\end{tikzcd}
\]
Given a CSS code $\csscode = (H_Z,H_X)$ we write the automorphisms
of the underlying classical code $\underlying$ as 
$$
\mathcal{G}:=\Aut(\underlying).
$$
The subgroup of $\mathcal{G}$ that fixes the $X$- and
$Z$-sectors corresponds precisely to $\Aut(\csscode)$ defined above.
In more detail, these elements of~$\mathcal{G}$ 
are pairs $(\sigma_n,\sigma_r)$ where $\sigma_r$
has block form
$$
\sigma_r = 
\left(
\begin{array}{cc}
\sigma_X & 0 \\
0 & \sigma_Z 
\end{array}
\right)
$$
and then $(\sigma_Z,\sigma_n,\sigma_X)$ is the corresponding
automorphism of $\csscode$.

The elements $\tau$ of $\mathcal{G}$ that swap the $X$- and
$Z$-sectors are called \emph{$ZX$-dualities}.
These are pairs $\tau=(\tau_n, \tau_r)$ such that
$\tau_r$ has the block form:
$$
\tau_r = \left( \begin{array}{cc} 0 & \tau_Z \\ \tau_X & 0 \end{array} \right).
$$
In this case, $\tau=(\tau_Z,\tau_n,\tau_X)$ is an isomorphism
$\tau:\csscode\to \csscode^\top,$
and we call such a code \emph{self-$ZX$-dual}.
The set of all $ZX$-dualities of~$\csscode$ we denote by~$\zxdualities$.

\begin{lemma}\label{lem:zx-duality}
Given a self-$ZX$-dual CSS code~$\csscode$,
the set of $ZX$-dualities for $\csscode$
is a coset of $\Aut(\csscode)$ in $\mathcal{G}$:
$$ \zxdualities = \tau \Aut(\csscode),$$
where $\tau$ is any choice of a $ZX$-duality for $\csscode$.
\end{lemma}
\begin{proof}
It is not hard to see that from a given $ZX$-duality $\tau:\csscode\to \csscode^\top$
and an isomorphism $\sigma:\csscode\to \csscode$ we get another $ZX$-duality from
the composition $\tau\sigma$.
Conversely,
given another $ZX$-duality $\tau'=(\tau_n',\tau_r')$
we find that
\begin{align*}
\tau_r^{-1} \tau'_r &=
\left( \begin{array}{cc} 0 & \tau_Z \\ \tau_X & 0 \end{array} \right)^{-1}
\left( \begin{array}{cc} 0 & \tau'_Z \\ \tau'_X & 0 \end{array} \right)\\
&=
\left( \begin{array}{cc}
\tau_X^{-1} \tau'_X & 0 \\
0 & \tau_Z^{-1} \tau'_Z 
\end{array} \right)
\in \Aut(\csscode).
\end{align*}
Therefore we have shown the identity
$\zxdualities = \tau \Aut(\csscode).$
\end{proof}

\begin{corollary}\label{cor:number_automorphisms}
For a self-$ZX$-dual code $\csscode$, 
the number of $ZX$-dualities 
equals the size of the automorphism group of $\csscode$: 
$$|\zxdualities| = |\Aut(\csscode)|.$$
\end{corollary}

Given CSS codes $C=(H_Z,H_X)$ and $C'=(H'_Z,H'_X)$
we can form the \emph{direct sum} CSS code:
$$
C\oplus C':= (H_Z\oplus H'_Z, H_X\oplus H'_X).
$$

\begin{lemma}\label{lem:stacking}
Given a CSS code $\csscode$ with $ZX$-duality
$\tau:\csscode\to\csscode^\top$ then
$\csscode\oplus\csscode$ has a $ZX$-duality
without any fixed points.
\end{lemma}
\begin{proof}
We define 
the $ZX$-duality given by $\tau\oplus\tau$
followed by the permutation that swaps the two summands:
$$
C\oplus C \xrightarrow{\tau\oplus\tau} C^\top\oplus C^\top
\xrightarrow{\mathrm{swap}}  C^\top\oplus C^\top = (C\oplus C)^\top.
$$
\end{proof}

In what follows we will often abuse notation
and write $\tau(i)$ for the action of $\tau_n$ on
the $i$-th basis vector of $\Field_2^n$.

\subsection{Hyperbolic surface codes}\label{sec:surface-codes}

In this section we review the construction of
hyperbolic triangle groups
(\cite{Girondo2012} Section 2.4), and how they
give rise to two dimensional hyperbolic surface codes~\cite{hyperbolic_codes,breuckmann2017homological}.
We then find the $ZX$-dualities in this context.

The Poincar\'e disc model of the hyperbolic plane 
is written as $\Disc = \{z\in\Complex : |z|<1\}.$
The group of orientation preserving automorphisms of
$\Disc$ is given by the M\"{o}bius transforms,
\begin{align*}
\Aut(\Disc^{+}) &= \Bigl\{ z \mapsto \frac{az + b}{cz + d} \ \Bigl|\  a,b,c,d\in\Real,\\
    & \ \ \ \ \  ad-bc=1 \Bigr\} \\
    & \cong \mathrm{PSL}(2,\Real).
\end{align*}
By also including the reflection given by
complex conjugation $z\mapsto \bar{z}$, we
generate the full automorphism group $\Aut(\Disc).$

\newcommand{\Rplus}{R^{+}_{5,5}}
\newcommand{\Gplus}{\mathcal{G}^{+}}
The Coxeter reflection group for the 
$\{5,5\}$-tiling of $\Disc$ has presentation
\begin{align*}
R_{5,5} := \langle a, b, c\ |\ a^2, b^2, c^2, (ab)^5, (ac)^2, (bc)^5 \rangle.
\end{align*}
The generators $a, b, c$ here correspond to reflections
that \emph{destabilize} a vertex, edge, or face respectively.
We can summarize this presentation using the Coxeter diagram:
\begin{center}
\includegraphics[]{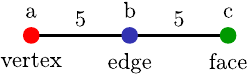}
\end{center}
See \Cref{fig:poincare-disc-525}.
The rotations $ab, ac, bc$ correspond to the rotation
around a face, edge or vertex respectively.
The subgroup generated by these rotations we
write as
$$
    \Rplus = \langle ab, ac, bc \rangle.
$$
This subgroup has index 2 in $R_{5,5}.$
Next we choose a normal subgroup $\Gamma$ of $\Rplus$
with finite index, that acts freely on $\Disc.$
The surface we are interested in is the compact Riemann surface
$$
    S = \Disc / \Gamma.
$$
This surface is tiled by triangles, which we can
identify with the elements of the group $\Gplus := \Rplus/\Gamma$
once we choose a fiducial tile to act upon.
Up to conjugation, the stabilizer subgroup of a face, edge or vertex
is generated by the rotations $ab, ac, bc \in \Gplus$, respectively.
Therefore, the chain complex of our quantum code is
obtained from the free vector spaces generated by cosets of these:
\begin{align}\label{eqn:hcode}
 C = \Bigl\{\ \vspac{r_Z} & := \Field_2[\Gplus / \langle ab \rangle] \nonumber \\
\xrightarrow{H_Z^\top} \vspac{n} &  := \Field_2[\Gplus / \langle ac \rangle] \nonumber \\
\xrightarrow{H_X} \vspac{r_X} &  := \Field_2[\Gplus / \langle bc \rangle]\ \Bigr\}.
\end{align}

Proceeding with the theory developed in \Cref{sec:ZX-dualities}
we now seek the automorphism group $\mathcal{G}$  of the underlying classical
code of $C$.
We can do this using another
Coxeter reflection group $R_{5,4}$ for the 
$\{5,4\}$-tiling of $\Disc$. This has presentation
$$
R_{5,4} := \langle d, e, f\ |\ d^2, e^2, f^2, (de)^5, (df)^2, (ef)^4 \rangle.
$$
The generators $d, e, f$ here correspond to reflections
that destabilize a vertex, edge, or face respectively.
The rotations $de, df, ef$ correspond to the rotation
around a face, edge or vertex respectively.
We can summarize this presentation using the Coxeter diagram:
\begin{center}
\includegraphics[]{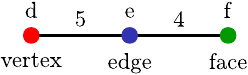}
\end{center}
See \Cref{fig:poincare-disc-524}.

\begin{lemma}\label{lem:coxeter_fold}
$R_{5,5}$ has an outer automorphism that swaps the
generators $a, c$ while fixing $b$.
Moreover, $R_{5,5}$ is a subgroup of $R_{5,4}$
in which the outer automorphism is conjugation by $f$.
\end{lemma}
\begin{proof}
The outer automorphism is clear from the presentation of $R_{5,5}$.
There is a group monomorphism 
$R_{5,5}\rightarrowtail R_{5,4}$
given by sending the generators 
$$
a \mapsto e,\ b \mapsto d,\ c \mapsto f^{-1}ef,
$$
and conjugation of these generators by $f$ has the required effect.
\end{proof}


\begin{theorem}\label{thm:hyperbolic}
For the hyperbolic code $C$ defined in \Cref{eqn:hcode} we have the following:
\begin{align*}
    \Aut(C) &= R_{5,5}/\Gamma, \\
    \Aut(\underlying) &= R_{5,4}/\Gamma.
\end{align*}
Furthermore, $C$ is self-$ZX$-dual as exhibited by
the $ZX$-duality $\tau=f\in R_{5,4}$.
\end{theorem}
\begin{proof}
(Sketch.)
By \Cref{lem:coxeter_fold} we can think of $R_{5,5}$ as a subgroup of $R_{5,4}$, where the latter has an additional generator~$f$.
The reflection~$f$ divides each of the fundamental triangles
of~$R_{5,5}$ with internal angles $\pi/2$, $\pi/5$, $\pi/5$
into two triangles with internal angles $\pi/2$, $\pi/4$,
$\pi/5$; see the shaded region in~\Cref{fig:poincare-disc-524,fig:poincare-disc-525}.

By studying the cosets which give rise to~$C$ and~$\underlying$ in the universal cover,
and noting that 
$\Gamma$ is normal and acts fixed-point free,
our arguments regarding the cosets are preserved under the covering map induced by
$\Gamma$ onto the finite surface.
\end{proof}

Using this result and \Cref{lem:zx-duality}
we obtain all the $ZX$-dualities of our hyperbolic code $C$
as
$$
    \zxdualities = f R_{5,5}/\Gamma.
$$

A similar result holds for other symmetric Coxeter
reflection groups $R_{l,l}$ for $l\ge 3$.
There is always a subgroup inclusion $R_{l,l}\rightarrowtail R_{l,4}$ 
coming from the outer automorphism of $R_{l,l}.$
In the case of $R_{3,3}$ this group is $\operatorname{S}_4$
the (finite) symmetry group of the tetrahedron, and $R_{3,4}$ is the
symmetry group of the cube.
The inclusion $R_{3,3}\rightarrowtail R_{3,4}$ corresponds to the
inclusion of a tetrahedron inside a cube.
When $l=4$, we get the group $R_{4,4}$ which is
the (infinite) symmetry group of the square two-dimensional lattice.
The inclusion  $R_{4,4}\rightarrowtail R_{4,4}$ exhibits
an inclusion of this lattice into a $1/\sqrt{2}$ scaled and rotated copy of the same lattice.

\begin{figure*}[t]
\centering
\begin{subfigure}[b]{0.3\textwidth}
\includegraphics[width=\textwidth]{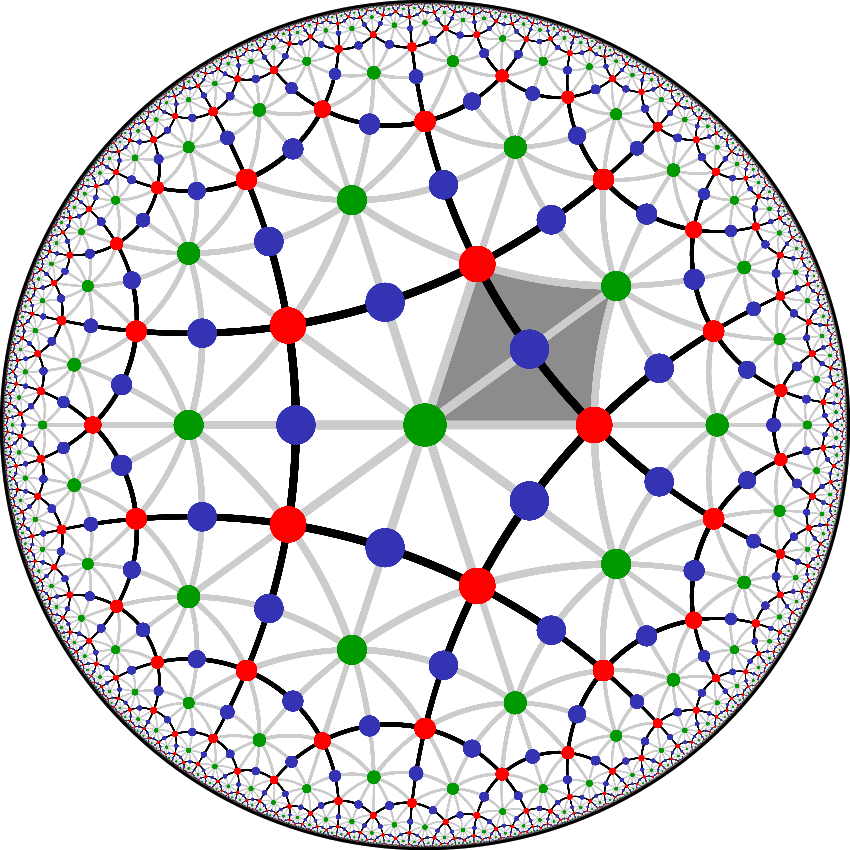}
\caption{
The $\{5,4\}$-tiling of the Poincar\'e hyperbolic disc model.
Vertices, edges and faces are marked with red, blue and green dots respectively.
This tiling has symmetry group the Coxeter reflection group $R_{5,4}$.
}
\label{fig:poincare-disc-524}
\end{subfigure}
~
\begin{subfigure}[b]{0.3\textwidth}
\includegraphics[width=\textwidth]{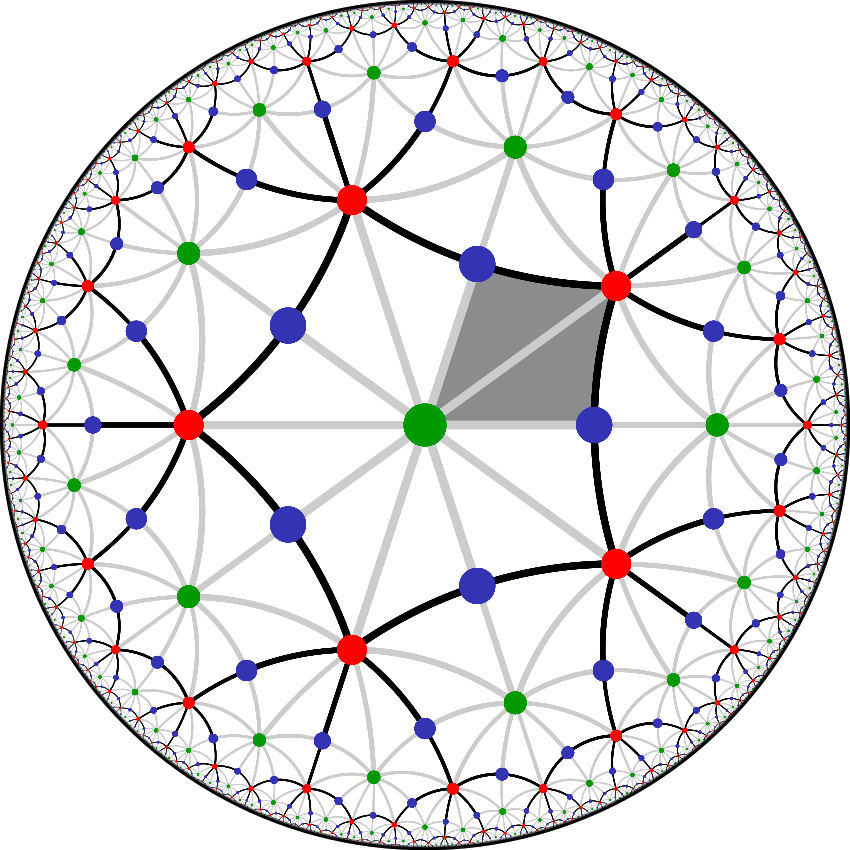}
\caption{
The $\{5,5\}$-tiling of the Poincar\'e hyperbolic disc model.
Vertices, edges and faces are marked with red, blue and green dots respectively.
This tiling has symmetry group the Coxeter reflection group $R_{5,5}$.
}
\label{fig:poincare-disc-525}
\end{subfigure}
~
\begin{subfigure}[b]{0.3\textwidth}
\includegraphics[width=\textwidth]{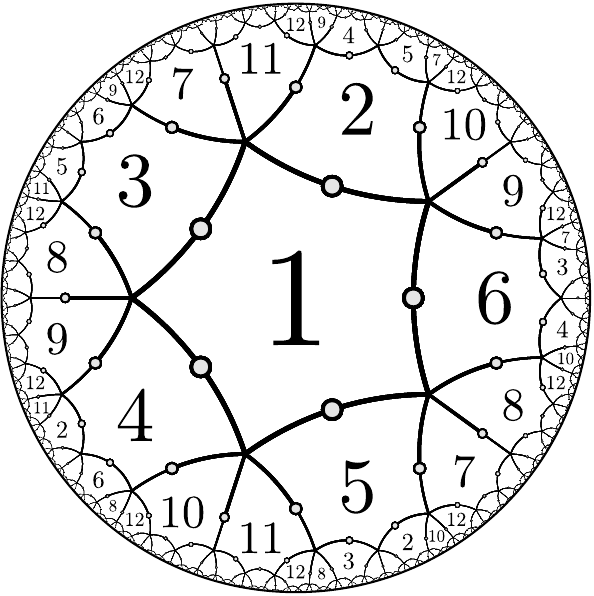}
\caption{
Bring's curve as a quotient of the $\{5,5\}$-tiled hyperbolic disc.  
Faces are identified according to the number scheme.
Each edge is associated with a qubit, shown with a grey circle,
and faces/vertices define $Z$/$X$-check operators.
}
\label{fig:brings-curve}
\end{subfigure}
\caption{
Hyperbolic codes are defined using the homology of
compact hyperbolic surfaces. We construct these tiled surfaces
as quotients of the $\{5,5\}$-tiling (ii), for example Bring's code (iii).
We find $ZX$-dualities of the hyperbolic codes by 
forgetting the distinction between vertices and faces in (ii), which yields
the $\{5,4\}$-tiling (i). The same quadrilateral region is shaded in (i) and (ii), 
demonstrating this relation between the Coxeter reflection groups.
}
\label{fig:poincare-disc}
\end{figure*}

\subsection{Fold-Transversal Clifford Gates}\label{sec:fold_transversal_gates}
Some of the $ZX$-dualities in $\zxdualities$ can be used
to implement logical gates using circuits of low, constant depth.

Given a $ZX$-duality $\tau$, a unitary operator
on the $n$ qubits is called \emph{fold-transversal}
when it is a tensor product of single qubit unitaries, and
two-qubit unitaries with support on the orbits 
$\{(i,\tau(i))\}_{i=1,...,n}$
of $\tau.$

A unitary operator on the $n$ qubits of a CSS code $\csscode$ 
is called an \emph{encoded logical gate} when it
commutes with the codespace projector.

\subsubsection{Hadamard-type fold-transversal gates}

The simplest such gate is of what we call \emph{Hadamard-type},
meaning it exchanges logical $X$-operators with logical $Z$-operators.
Given a $ZX$-duality $\tau \in \zxdualities$ we define
this as the following fold-transversal operator on the $n$ physical bits:
\begin{align*}
	H_\tau = \bigotimes_{\substack{i=1,..,n\\i<\tau(i)}} 
        \operatorname{SWAP}_{i,\tau(i)}\ \bigotimes_{i=1}^n H_i 
\end{align*}
\begin{theorem}
Given a $ZX$-duality $\tau$ for a CSS code $\csscode$
then $H_\tau$ is an encoded logical gate.
\end{theorem}
\begin{proof}
The operator $H_\tau$ acts on elements of the Pauli group $\mathcal{P}_n$ via conjugation.
Using the identities $H_i X_i H_i^\dagger = Z_i$ and $H_i Z_i H_i^\dagger = X_i$, as well as the fact that $\tau$ maps between $X$- and $Z$-checks, we see that $H_\tau$ leaves the stabilizer group $\mathcal{S}$ invariant.
Thus, the normalizer of the stabilizer group~$N(\mathcal{S})$ must also be invariant.
In other words, $H_\tau$ leaves the code space invariant and maps logical operators onto logical operators.
\end{proof}

Note that $H_\tau$ is generally not simply a logical Hadamard, exchanging pairs of logical Pauli-operators~$\bar{X}_i$ with~$\bar{Z}_i$, as~$\tau$ does not have to be compatible with the choice of logical basis.
More severely, in general, we do not have $H_\tau^2 = I$.
This is only the case if $\tau^2 = \id$.

When we construct $H_\tau$ for the surface code, using the $ZX$-duality $\tau$ shown in \Cref{fig:fold-surface}, we obtain the Hadamard-gate described by Moussa~\cite{moussa2016transversal}.

\subsubsection{Phase-type fold-transversal gates}\label{sec:phase_type_gates}

Given a $ZX$-duality $\tau\in\zxdualities$ we define the following
fold-transversal operator on the $n$ physical qubits:
\begin{align*}
S_\tau = \bigotimes_{\substack{i=1,...,n\\i=\tau(i)}} S_i^{(\dagger)}\ 
    \bigotimes_{\substack{i=1,...,n\\i < \tau(i)}} \operatorname{CZ}_{i,\tau(i)}
\end{align*}
for some choice of $S_i$ and $S_i^\dagger$ on the fixed points $i=\tau(i).$
We call this a \emph{phase-type} gate.

\begin{theorem}\label{thm:S_type_gate}
Let $\tau \in \zxdualities$ be a $ZX$-duality for a CSS code $\csscode$,
which is
self-inverse ($\tau^2=\id$) and whose set of stabilized
qubits is even, i.e. $|\{ i=1,...,n \mid \tau(i)=i \}| \equiv 0 \mod 2$.
Further, suppose that each $X$-check~$X^{\otimes s}$
with support $s\subset \{1,\dotsc,n\}$ has (a) even overlap
with the set of invariant qubits $\{i=1,...,n  \mid \tau(i)=i\}$ and 
(b) contains an even number of two element orbits 
$\{i,\tau(i)\}$, where $i\neq \tau(i)$.
Then there exists a phase-type gate~$S_\tau$ that is
an encoded logical gate for $\csscode$.
\end{theorem}
\begin{proof}
Due to condition (a) there exists a partition of the
set of invariant qubits of~$\tau$ into~$A$ and~$B$ such
that each $X$-check is half supported on~$A$ and half supported on~$B$.
We define the following unitary operator on the $n$ physical bits:
\begin{align*}
S_\tau = \bigotimes_{i\in A} S_i\ \bigotimes_{j\in B} S_j^\dagger\ 
    \bigotimes_{\substack{i=1,...,n\\i < \tau(i)}} \operatorname{CZ}_{i,\tau(i)}
\end{align*}
Let us verify that~$S_\tau$ leaves the stabilizer group~$\mathcal{S}$ invariant.
We do this by showing that~$S_\tau$ sends check operators to other elements
of~$\mathcal{S}$. By virtue of~$S_\tau$ being Clifford, this then extends to
a permutation of~$\mathcal{S}$.
We first note that as~$S_\tau$ is diagonal in the computational basis, it commutes with all $Z$-checks.
Consider an $X$-check $X^{\otimes s}$ with support $s\subset \{1,\dotsc,n\}$.
First, let us assume that~$s$ does not contain qubits invariant under~$\tau$, i.e.~$s\cap (A\cup B) = \emptyset$.
Note that for single-qubit Pauli-$X$ operators we have that $$\operatorname{CZ}_{i,j} X_i \operatorname{CZ}_{i,j}^\dagger = X_i Z_j .$$
For any pair of qubits $i$ and $j=\tau(i)$ we have that $$CZ_{i,j} X_i X_j CZ_{i,j} =
- X_i X_j Z_i Z_j .$$
Due to condition (b) the total number of minus signs that we pick up is even.
Hence, the operator $X^{\otimes s}$ is mapped by $S_\tau$ onto the operator $X^{\otimes s} Z^{\otimes \tau(s)}$.
By virtue of~$\tau$ mapping $X$-checks onto $Z$-checks, we know that $Z^{\otimes \tau(s)}\in \mathcal{S}$.
Now, let us assume that $s\cap (A\cup B) \neq \emptyset$.
For Pauli-$X$ operators we have that $S_i X_i S_i^\dagger = i X_iZ_i$ and $S_i^\dagger X_iS_i = -i X_i Z_i$.
The check~$X^{\otimes s}$ is mapped onto $$i^{|s\cap A|}(-i)^{|s\cap B|} X^{\otimes s} Z^{\otimes \tau(s)}.$$
Each phase gate $S_i$ ($S_i^\dagger$) introduces an unwanted factor of $i$ ($-i$).
However, by assumption on $\tau$, we are guaranteed that $|s\cap A| = |s\cap B|$, so that these factors cancel out.
\end{proof}

As for the Hadamard-type gates, we have shown that the phase-type gate $S_\tau$ leaves the code space invariant and maps logical operators onto logical operators.
We stress again that the action of $S_\tau$ will generally be different from applying a logical $\bar{S}_i$ on each logical qubit.
However, we do have $S_\tau^4 = I$ by construction. 

When we construct $S_\tau$ for the surface code, using the $ZX$-duality $\tau$ shown in \Cref{fig:fold-surface}, we obtain the $S$-gate described by Moussa~\cite{moussa2016transversal}.
Note there is another $ZX$-duality that rotates the surface code by 90 degrees.
However, this rotation fixes a single physical qubit
and so we cannot construct a phase-type encoded logical gate from this $ZX$-duality.

The next result shows we can get a phase-type encoded logical
gate by \emph{stacking} two copies of a self-$ZX$-dual code.

\begin{corollary}\label{thm:stacking}
Let $\tau \in \zxdualities$ be a $ZX$-duality for a CSS code $\csscode$ which is self-inverse,
then $\csscode\oplus\csscode$ has phase-type encoded logical gate:
\begin{align*}
S_\tau = \bigotimes_{i=1,...,n} \operatorname{CZ}_{i,n+\tau(i)}.
\end{align*}
\end{corollary}
\begin{proof}
We use the $ZX$-duality defined in \Cref{lem:stacking},
and check that this duality satisfies the condition of the
previous theorem.
Indeed, this $ZX$-duality has no fixed points, is self-inverse,
and the orbits intersect the support of an $X$-check on at most one qubit.
\end{proof}

	\subsection{The Symplectic Group}
	Instead of dealing with unitary Clifford operators directly, we consider a well-known relation to the \emph{symplectic group}~$\operatorname{Sp}_{2n}(\mathbb{F}_2)$.
	
	By definition, we have that the Pauli group~$\mathcal{P}_n$ is a normal subgroup of the Clifford group~$\mathcal{C}_n$.
	The quotient group~$\mathcal{C}_n / \mathcal{P}_n$ is isomorphic to the symplectic group~$\operatorname{Sp}_{2n}(\mathbb{F}_2)$~\cite{bolt1961clifford}.
	In other words, the following short sequence is exact:
	$$
	\begin{tikzcd}
		\mathcal{P}_n  \arrow[r, tail, "\iota"] & \mathcal{C}_n  \arrow[r, two heads, "\pi"] & \operatorname{Sp}_{2n}(\mathbb{F}_2)
	\end{tikzcd}
	$$
	where~$\iota$ is the canonical embedding map and~$\pi$ the quotient map.
	This implies that any element of the Clifford group $U\in \mathcal{C}_n$ can be specified by an element of the symplectic group $\pi(U)\in \operatorname{Sp}_{2n}(\mathbb{F}_2)$ and a representative of the Pauli group $P\in \mathcal{P}_n$.
	This fact is well-known and used e.g. for sampling random elements of the Clifford group~\cite{koenig2014efficiently}.
	Note that we can trivially implement logical Pauli operators $\bar{P}\in \mathcal{P}_k$ transversally.
	
	From now on, we will use the notation $\tilde{U} := \pi(U)$ for any operator $U\in \mathcal{C}_n$. 
	
%

	\section{Bring's code}\label{sec:brings_code}
	
	We demonstrate our construction on a hyperbolic surface code, called Bring's code, which encodes 8 logical qubits into 30 physical qubits.
	Hyperbolic codes are good candidates for constructing interesting fold-transversal gates, as they are finite-rate and have large symmetry groups~\cite{hyperbolic_codes}.
	One reason for choosing Bring's code in particular is that its symmetries can be visualized as a 3D polyhedron (see~\Cref{fig:great_dodecahedron}).
	Also, it is small enough to analyze the group generated by the fold-transversal gates with a computer algebra system, such as \textsc{GAP}~\cite{GAP4}.
	Larger hyperbolic codes already encode too many logical qubits to do so.
	
	We briefly discuss some further examples, such as hypergraph product codes, balanced product codes and block codes, in \Cref{sec:further_examples}.
	
	\subsection{Construction}\label{sec:construction_bring}

Using the notation from \Cref{sec:surface-codes} 
for Bring's code we define $\Gamma$ equal to the
normal closure in $R^{+}_{5,5}$ of the group generated by $(abcb)^3.$
This single generator is enough to define the quotient of the hyperbolic disc $\Disc$,
see \Cref{fig:brings-curve}.
Somewhat surprisingly, this quotient has an 
immersion into three dimensional space known as the \emph{great dodecahedron}.
This is a Kepler-Poinsot polyhedron derived from the dodecahedron, see \Cref{fig:great_dodecahedron}.
It has pentagonal faces, of which five meet at any vertex.
By identifying the edges of the great dodecahedron with qubits and associating vertices/faces with \mbox{$X$-/}$Z$-checks, we obtain a quantum code with parameters $[[30,8,3]]$.
	
	\begin{figure}
		\centering
		\includegraphics[width=0.8\linewidth]{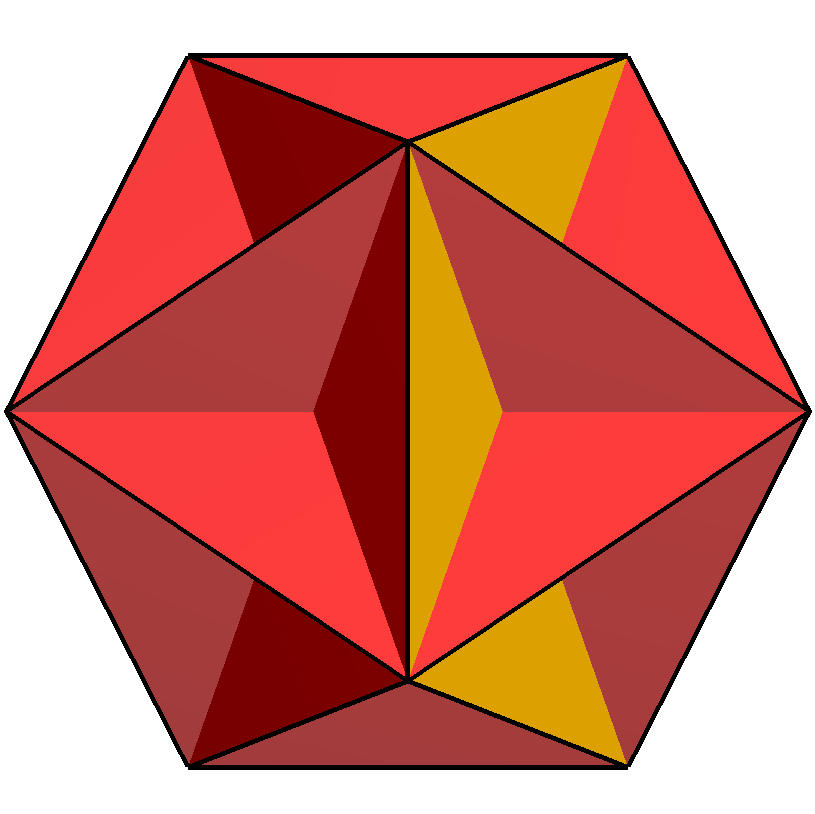}
		\caption{
			The great dodecahedron representation of Bring's code.
			Faces are pentagons which intersect.
			A single pentagonal face is highlighted yellow,
and these correspond to the pentagons in \Cref{fig:brings-curve}.
			Each edge is associated with a qubit and vertices/faces define $X$/$Z$-check operators.
			This figure was made with \textsc{Stella4D}~\cite{webb2003stella}.
		}
		\label{fig:great_dodecahedron}
	\end{figure}
	
From a topological perspective, the great dodecahedron is a tiling (cellulation) of a genus-4 surface,
$S=\Disc/\Gamma.$
This tiling consists of 12 vertices, 30 edges and 12 faces.
We refer to the derived code as \emph{Bring's code},
as the underlying surface is known as \emph{Bring's curve},
named after Erland Samuel Bring, who studied the algebraic
equation that defines the surface as a complex curve.
Bring's code was first defined in~\cite{hyperbolic_codes}
and further analyzed in~\cite{hyperbolic_quantum_storage} and~\cite{conrad2018small}.
	
%

From \Cref{thm:hyperbolic}, the group of symmetries of Bring's code 
is $\Aut(C)=R_{5,5}/\Gamma$ which in this case is found to be the
permutation group~$\operatorname{S}_5$.
In terms of the presentation of~$R_{5,5}$ 
we have,
\begin{align*}
\Aut(C) & = R_{5,5}/\Gamma \\
& = \langle a,b,c \mid  a^2, b^2, c^2, \\
    &\ \ \ \ \  (ab)^5, (ac)^2, (bc)^5, (abcb)^3 \rangle \\
 &= \operatorname{S}_5.
\end{align*}
We also make the definitions $r:=ab$ and $s:=bc$ which are rotations
around a face and vertex, respectively.
Also from \Cref{thm:hyperbolic}, we have the
larger symmetry group which includes symmetries
that exchange vertices and faces,
\begin{align*}
\Aut(\underlying) &= R_{5,4}/\Gamma \\
    &= \operatorname{S}_5\times\operatorname{C}_2.
\end{align*}

	\subsection{Logical bases}\label{sec:bases}
	In the homological representation of CSS codes, the logical $Z$-operators correspond to a basis of the homology group~$\Homology_1$.
	Essential cycles on Bring's surface can be associated to the 20 faces of the icosahedron, which forms the convex hull of the great dodecahedron, see \Cref{fig:great_dodecahedron}.
	First we note that a chain consisting of three edges of a triangle is a closed loop and hence it is clearly a cycle, i.e. an element of $Z_1 = \ker \partial_1$.
	Direct computation shows that it is also not a boundary, i.e. an element of $B_1 = \im \partial_2$.
	In fact, the set of all 20 chains associated to the triangles span the first homology group~$\Homology_1$.
	However, as $k = \dim \Homology_1 = 8$, this spanning set can not be independent.
	Picking a suitable subset of these triangles gives rise to a basis.
	
	\begin{figure}
		\centering
		\includegraphics[width=0.8\linewidth]{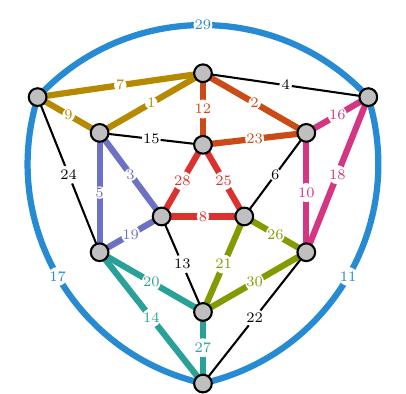}
		\caption{The basis $\mathfrak{B}_{\Homology_1}$ of the first homology group~$\Homology_1$.
			Each element is a cycle of length three shown with thick lines of one color.
			Equivalently, this is a choice of logical $Z$-operators $\bar{Z}_1,\dotsc,\bar{Z}_8$.
			Each triangular face of this planar embedding (of which there are 20, including the outer face) corresponds to a face of the convex hull of the great dodecahedron, which is the icosahedron.}
		\label{fig:logicals_planar}
	\end{figure}
	
	A particularly nice basis is shown in \Cref{fig:logicals_planar}.
	The basis elements correspond to eight faces of the icosahedron.
	These faces are related by the action of the subgroup $\operatorname{C}_2\times \operatorname{C}_2\times \operatorname{C}_2 \subset \operatorname{S}_5$.
	
	\begin{figure}
		\centering
		\includegraphics[width=0.8\linewidth]{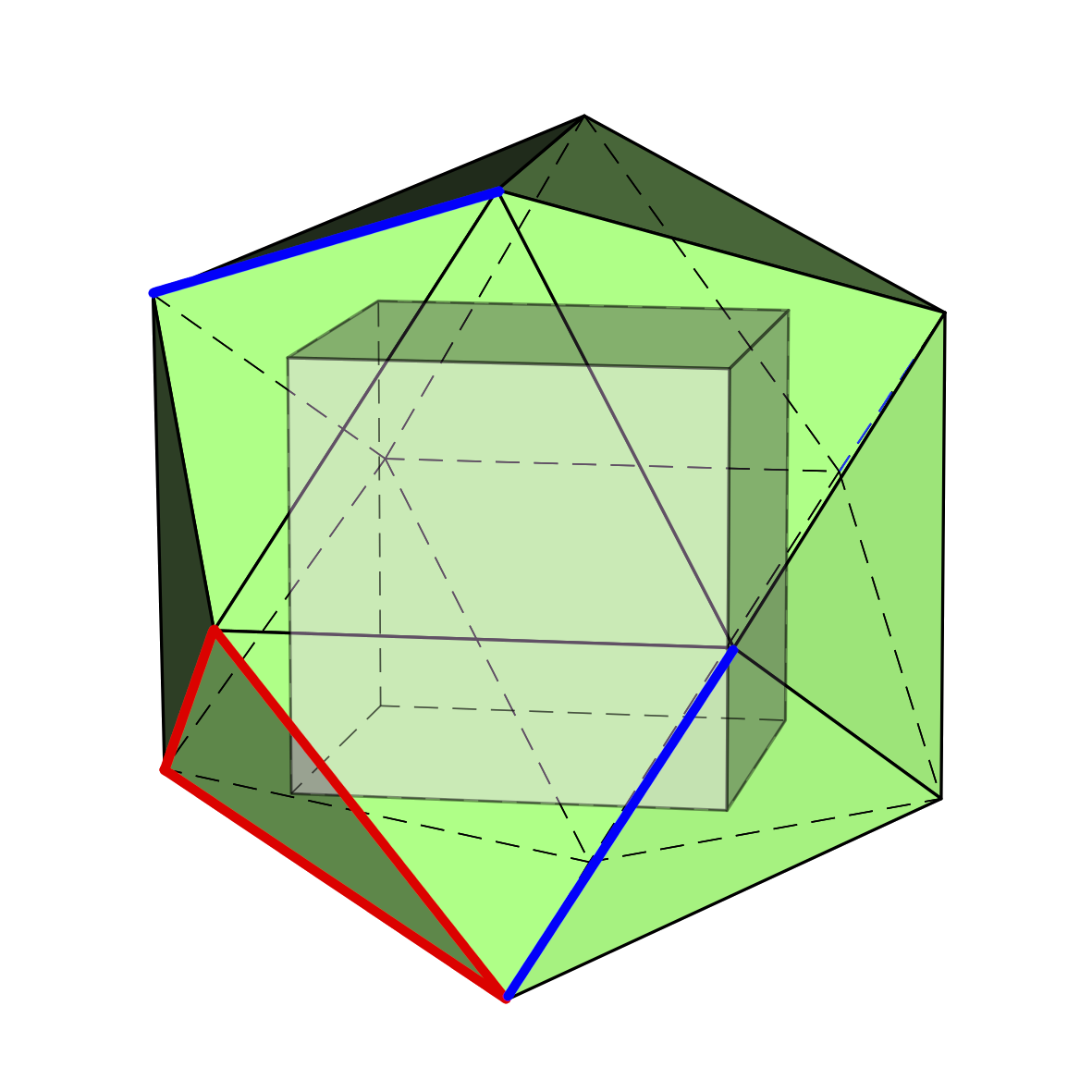}
		\caption{The convex hull of the great dodecahedron is the icosahedron.
			We can place a cube inside the icosahedron such that the eight vertices of the cube lie on the mid-points of eight faces of the icosahedron.
			The symmetries of the cube are also symmetries of the icosahedron.
			The orientation of the cube inside the icosahedron corresponds to a choice of basis in~$\Homology_1$ where each triangle touched by a vertex of the cube is a basis vector, cf.~\Cref{fig:logicals_planar}.
			The basis element of $\Homology_1$ associated with the front, lower-left corner of the cube is highlighted in red and the basis element  of $\Homology^1$ associated with the front, upper-right corner of the cube is highlighted in blue.
		}
		\label{fig:icosahedron_cube}
	\end{figure}
	This choice of basis has a geometric representation that can be seen in \Cref{fig:icosahedron_cube}.
	A single basis element of the homology group is highlighted in red.
	We can place a cube inside the icosahedron, such that eight faces of the icosahedron can be associated with the eight vertices of the cube.
	The group $\operatorname{C}_2\times \operatorname{C}_2\times \operatorname{C}_2$ is a subgroup of the symmetry group of the cube which has a regular action on its vertices.
	It is generated by reflections corresponding to the three planes intersecting the cube in the middle.
	This placement of the icosahedron inside the cube is not unique and different placements correspond to different choices of a basis.
	
	The logical $X$-operators correspond to a basis of the cohomology group~$\Homology^1$.
	Similarly as for the homology basis, we can associate the essential cocycles with triangles of the icosahedron:
	for a given triangle there are three edges (not belonging to the triangle) which are adjacent to exactly one of the triangle's vertices and are not neighboring any of the triangle's edges (see \cref{fig:icosahedron_cube}, highlighted in blue).
	The chain formed by these edges is an essential cocycle.
	By associating eight triangles with the vertices of a cube we obtain a basis of~$\Homology^1$.
	Unfortunately, choosing the bases of $\Homology_1$ and~$\Homology^1$ in this fashion can not give rise to a symplectic basis of $\Homology_1 \oplus \Homology^1$.
	
	In order to represent the logical gates that we are going to construct in terms of symplectic matrices, we need to fix ordered bases of~$\Homology_1$ and~$\Homology^1$.
	We fix the basis~$\mathfrak{B}_{\Homology_1}$ of~$\Homology_1$ as shown in \Cref{fig:logicals_planar} with the order
	$\{1,7,9\}$,
	$\{2,12,23\}$,
	$\{3,5,19\}$,
	$\{10,16,18\}$,
	$\{21,26,30\}$,
	$\{8,25,28\}$,
	$\{11,17,29\}$,
	$\{14,20,27\}$.
	The ordered basis elements of~$\mathfrak{B}_{\Homology^1}$, given in terms of supports on the edges (cf.~\Cref{fig:logicals_planar}), are
	$\{1,6,11\}$,
	$\{2,8,22\}$,
	$\{4,5,25\}$,
	$\{10,13,17\}$,
	$\{24,28,30\}$,
	$\{12,19,26\}$,
	$\{9,16,27\}$,
	$\{15,20,29\}$.
	
	The bases $\mathfrak{B}_{\Homology_1}$ and $\mathfrak{B}_{\Homology^1}$ give rise to a basis  $\mathfrak{B}_{\Homology_1 \oplus \Homology^1}$ of the symplectic vector space~$\Homology_1 \oplus \Homology^1$ in the obvious way.
	The symplectic product between the elements of~$\mathfrak{B}_{\Homology_1 \oplus \Homology^1}$ is given by the matrix $\Phi$ with $\Phi_{i,j} = \langle b^i,b_j \rangle$ with $b^i\in \mathfrak{B}_{\Homology^1}$ and $b_j\in \mathfrak{B}_{\Homology_1}$.
	\begin{align*}
		\Phi = 
		\left(\begin{array}{rrrrrrrr}
		1 & 0 & 0 & 0 & 0 & 0 & 1 & 0 \\
		0 & 1 & 0 & 0 & 0 & 1 & 0 & 0 \\
		0 & 0 & 1 & 0 & 0 & 1 & 0 & 0 \\
		0 & 0 & 0 & 1 & 0 & 0 & 1 & 0 \\
		0 & 0 & 0 & 0 & 1 & 1 & 0 & 0 \\
		0 & 1 & 1 & 0 & 1 & 0 & 0 & 0 \\
		1 & 0 & 0 & 1 & 0 & 0 & 0 & 1 \\
		0 & 0 & 0 & 0 & 0 & 0 & 1 & 1
		\end{array}\right)
	\end{align*}
	As $\Phi$ is not diagonal, $\mathfrak{B}_{\Homology_1 \oplus \Homology^1}$ is not a symplectic basis.

	We can fix this by defining alternative bases as follows:
	Let $b^{'}_i = b_i$ for $i\in \{1,\dotsc,5\}$, $b^{'}_6 = b_2 + b_3 + b_5 + b_6$, $b^{'}_7 = b_1 + b_4 + b_7$ and $b^{'}_8 = b_1 + b_4 + b_7 + b_8$, defining a new basis~$\mathfrak{B}^{'}_{\Homology_1}$ of the first homology group.
	Further, let $b^{'i} = b^i$ for $i\in \{1,\dotsc,5\}$, $b^{'6} = b^2 + b^3 + b^5 + b^6$, $b^{'7} = b^8$ and $b^{'8} = b^1 + b^4 + b^7$, defining a new basis $\mathfrak{B}^{'}_{\Homology^1}$ of the first cohomology group.
	The resulting basis $\mathfrak{B}^{'}_{\Homology_1 \oplus \Homology^1}$ of the combined symplectic space~$\Homology_1 \oplus \Homology^1$ is a symplectic basis.
	
	\subsection{Permutation gates}
	Before we discuss the fold-transversal gate constrictions, we note that some non-trivial gates can be implemented simply by permuting qubits~\cite{grassl2013leveraging,Zeng2011}.
	
	By definition, the automorphisms of Bring's code directly correspond to the symmetries of Bring's surface.
	The subgroup of symmetries of Bring's surface, which maps vertices to vertices and faces to faces, is~$\operatorname{S}_5$ and acts transitively on the 30 edges.
	This action is also faithful, so that we can write the generators $a$, $b$ and $c$ of $\operatorname{S}_5$ in terms of cycles of $\operatorname{S}_{30}$, i.e. cycles permuting the edge labels:
	\begin{widetext}
		\begin{align*}
		a = (2,3)(4,5)(6,8)(7,9)(10,13)(11,14)(12,15)(16,19)(18,20)(21,26)(22,27)(23,28)(24,29)\\
		b = (1,2)(3,6)(4,7)(5,10)(9,16)(11,17)(13,21)(14,22)(15,23)(18,24)(19,26)(20,30)(25,28)\\
		c = (2,4)(3,5)(6,11)(7,12)(8,14)(9,15)(10,18)(13,20)(17,25)(21,27)(22,26)(23,29)(24,28)
		\end{align*}
	\end{widetext}
	This action of $\operatorname{S}_5$ on the edges extends to a linear representation $\rho$ on the space of 1-chains.
	
	Since clearly~$Z_1$ and~$B_1$ are invariant subspaces under any permutation matrix in the image of $\rho$, there exists a representation $\sigma_{\Homology_1} : \operatorname{S}_5 \to \operatorname{GL}(H_1)$ on the first homology group.
	Similarly, $Z^1$ and~$B^1$ are invariant subspaces as well, so that we obtain a representation $\sigma_{\Homology^1} : \operatorname{S}_5 \to \operatorname{GL}(H^1)$ on the first cohomology group.
	See \Cref{sec:permutation_rep} for the matrix representations.

	The group of logical gates induced by permutations of the edges is generated by
	\begin{align*}
	\tilde{\sigma}_x = 
	\left(
	\begin{array}{c|c}
	\tilde{\sigma}_{\Homology_1}(x) & 0 \\ \hline
	0 & \tilde{\sigma}_{\Homology^1}(x)
	\end{array}\right)
	\end{align*}
	where $x \in \{ a,b,c \}$.
	
	As the representation is faithful, we have that $\langle \tilde{\sigma}_a, \tilde{\sigma}_b, \tilde{\sigma}_c \rangle \cong \operatorname{S}_5  < \operatorname{Sp}_{16}(\mathbb{F}_2)$.
	This can be verified using a computer algebra system, such as \textsc{GAP}~\cite{GAP4}.

	\subsection{$ZX$-Dualities}
The group of automorphisms of Bring's code is~$\Aut(C)=\operatorname{S}_5$ which has order $5! = 120$.
Furthermore, Bring's code is self-ZX-dual, by \Cref{thm:hyperbolic},
with self-ZX-duality
$$\tau_0 := f\in R_{5,4}.$$
By \Cref{cor:number_automorphisms} we know that there
are 120 $ZX$-dualities exchanging $X$- and $Z$-checks of Bring's code.
Among these, the number of self-inverse dualities ($\tau^2 = \id$) is~20.
For the implementation of the Hadamard- and phase-type
gates in \Cref{sec:fold_transversal_bring} we will consider
the $ZX$-duality $\tau_0$.

	\Cref{fig:duality} shows the orbits of edges under the action of~$\tau_0$.
	As~$\tau_0$ is self-inverse each edge is either mapped onto a partner ($\tau_0(e_1)=e_2$ and $\tau_0(e_2)=e_1$), or it is fixed ($\tau_0(e)=e$).
	In the figure, the pairs are given the same color, while the fixed edges are drawn in black.
	The same information is shown in \Cref{fig:duality_planar} in the planar layout.
	
	We have picked the ordered bases $\mathfrak{B}_{\Homology_1}$ and~$\mathfrak{B}_{\Homology^1}$ in \Cref{sec:bases} such that $\tau_0$ maps between the pairs of basis elements, i.e. $\tau_0(b_i) = b^i$ for $b_i\in \mathfrak{B}_{\Homology_1}$ and $\tau_0(b^i) = b_i$  for $b^i\in \mathfrak{B}_{\Homology^1}$.

	\begin{figure}[t]
		\centering
		\includegraphics[width=0.8\columnwidth]{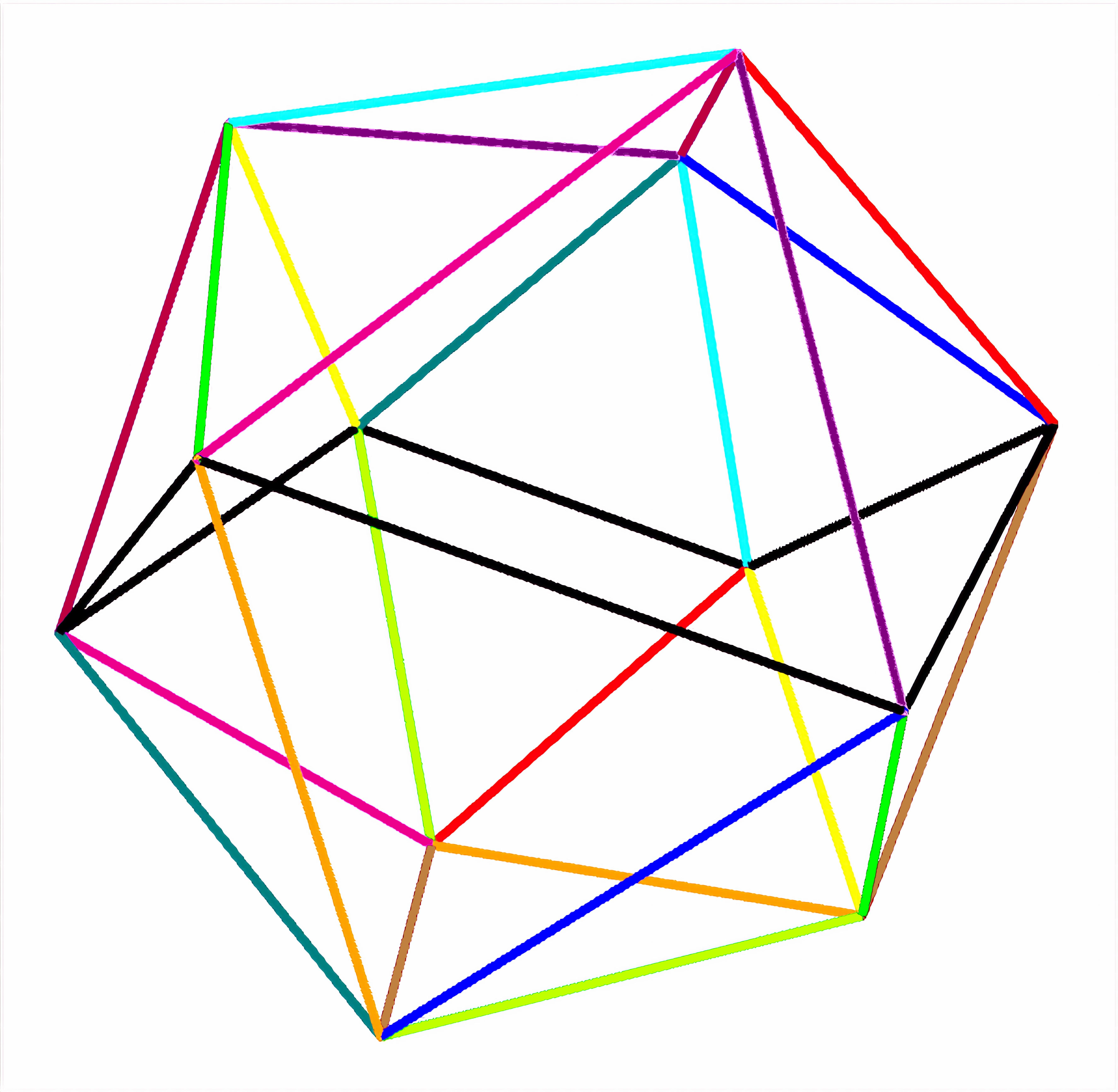}
		\caption{The $ZX$-duality $\tau_0$ on Bring's code. Qubits are associated with edges. Qubits colored in black are left invariant under the $ZX$-duality. Qubits not left invariant come in pairs, as $\tau_0^2=\id$, and are given the same color.}
		\label{fig:duality}
	\end{figure}

	\begin{figure}[t]
		\centering
		\includegraphics{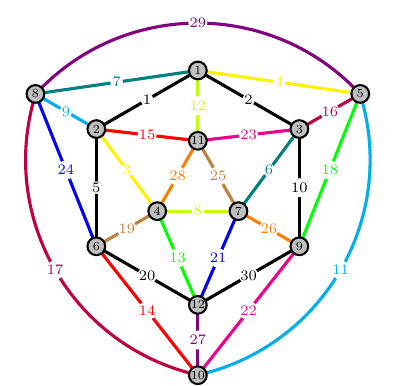}
		\caption{The orbits of $\tau_0$ shown in the planar layout.}
		\label{fig:duality_planar}
	\end{figure}

	\subsection{Fold-transversal gates}\label{sec:fold_transversal_bring}
	We will now use the $ZX$-duality $\tau_0$ from the previous section to construct fold-transversal Clifford gates for Bring's code.
	This is straightforward, as $\tau_0$ and the bases of~$\Homology_1 \oplus \Homology^1$ were carefully chosen.

	\subsubsection{Hadamard-type gate}\label{sec:hadamard}
	We have picked the bases $\mathfrak{B}_{\Homology_1}$ and $\mathfrak{B}_{\Homology^1}$ such that $\tau_0(b_i) = b^i$ and $\tau_0(b^i) = b_i$ for each $b_i\in \mathfrak{B}_{\Homology_1}$ and $b^i\in \mathfrak{B}_{\Homology^1}$.
	Therefore, the representation of $\tilde{H}_{\tau_0}$ in the basis of~$\mathfrak{B}_{\Homology_1 \oplus \Homology^1}$ takes the simple form
	\begin{align*}
		\tilde{H}_{\tau_0} = 
		\left(
		\begin{array}{c|c}
		0 & I_8 \\ \hline
		I_8 & 0
		\end{array}\right).
	\end{align*}
	Changing to the symplectic basis $\mathfrak{B}_{\Homology_1 \oplus \Homology^1}^{'}$ yields
	\begin{align*}
		\sbox0{$\begin{matrix}I_6& & \\ &0&1\\ &1&1\end{matrix}$}
		\sbox1{$\begin{matrix}I_6& & \\ &1&1\\ &1&0\end{matrix}$}
		{\tilde{H}}_{\tau_0}^{'} = 
		\left(
		\begin{array}{c|c}
		\makebox[\wd0]{\huge $0$} & \usebox{0}\\
		\hline
		\usebox{1} & \makebox[\wd0]{\huge $0$}
		\end{array}
		\right) .
	\end{align*}

	\subsubsection{Phase-type gate}\label{sec:phase}
	First we note that $\tau_0$ fulfills all requirements listed in \Cref{sec:phase_type_gates}:
	it is self-dual, the number of stabilized qubits is six, there are either none or two stabilized qubits in the support of each $X$-check and no $X$-check contains an orbit of cardinality two.
	All of this can be verified by inspecting \Cref{fig:duality_planar}.
	
	The representation of~$\tilde{S}_{\tau_0}$ in the basis~$\mathfrak{B}_{\Homology_1 \oplus \Homology^1}$ takes the simple form
	\begin{align*}
	\tilde{S}_{\tau_0} = 
	\left(
	\begin{array}{c|c}
	I_8 & I_8 \\ \hline
	I_8 & 0
	\end{array}\right).
	\end{align*}
	Changing to the symplectic basis $\mathfrak{B}_{\Homology_1 \oplus \Homology^1}^{'}$ yields
	\begin{align*}
	\sbox1{$\begin{matrix}I_6& & \\ &1&1\\ &1&0\end{matrix}$}
	{\tilde{S}}_{\tau_0}^{'} = 
	\left(
	\begin{array}{c|c}
	\makebox[\wd1]{\large $I_8$} & \makebox[\wd1]{\large $0$}\\
	\hline
	\usebox{1} & \makebox[\wd1]{\large $I_8$}
	\end{array}
	\right) .
	\end{align*}

	\subsection{Achievable gates}
	Let $G_{\tau_0}$ be the group generated by the permutation gates and fold-transversal gates, i.e. it is generated by the set $\{\tilde{\sigma}_a, \tilde{\sigma}_b, \tilde{\sigma}_c, \tilde{H}_{\tau_0}, \tilde{S}_{\tau_0}\}$.
	We observe, using \textsc{GAP}~\cite{GAP4}, that the group~$G_{\tau_0}$ is isomorphic to $\operatorname{Sp}_8(\mathbb{F}_2) \times \operatorname{C}_2$.
	We can also verify that~$G_{\tau_0}$ does not depend on the choice of self-inverse $ZX$-duality~$\tau_0$.
	Further, adding any other Hadamard-type or phase-type gate does not enlarge the generated group.
	
	\subsubsection{Full Clifford group}
	Adding the symplectic equivalent of a single $\operatorname{CNOT}$- or $\operatorname{CZ}$-gate between any of the logical qubits to the set $\{\tilde{\sigma}_a, \tilde{\sigma}_b, \tilde{\sigma}_c, \tilde{H}_{\tau_0}, \tilde{S}_{\tau_0}\}$ elevates it to a generating set of~$\operatorname{Sp_{16}(\mathbb{F}_2)}$.
	Hence, we can obtain the full Clifford group~$\mathcal{C}_8$.
	
	This can be done in several ways:
	In \cite{hyperbolic_quantum_storage} the authors show how to use Dehn twists to implement logical $\operatorname{CNOT}$-gates.
	Their technique requires either an increase of the qubit degree\footnote{Meaning the cumulative total number of other qubits a single qubit has to be connected with. The instantaneous qubit degree is constant throughout the procedure.} up to $O(\log(n))$ or, alternatively, an increase in overhead due to ancilla qubits.
	The temporal overhead of this technique is $O(d^2)=O(\log(n)^2)$.
	Alternatively, the authors of~\cite{cohen2021low} suggest a generalization of lattice surgery facilitated by an ancillary surface code.
	This results in an overhead in ancilla qubits scaling as~$O(d^2)$.
	
	\subsubsection{Fold-transversal Clifford group}
	Implementing entangling gates can certainly be done using lattice surgery and code deformation techniques.
	However, it would be much more appealing to achieve all logical Clifford gates fold-transversally.
	It turns out that this can indeed be done:
	Let $\tilde{\sigma}_r = \tilde{\sigma}_a \tilde{\sigma}_b$ and $\tilde{\sigma}_s = \tilde{\sigma}_b \tilde{\sigma}_c$.
	
	Let $\tilde{\sigma}_r = \tilde{\sigma}_a \tilde{\sigma}_b$ and $\tilde{\sigma}_s = \tilde{\sigma}_b \tilde{\sigma}_c$.
	Under the action of~$G_{\tau_0}$ the symplectic space $\Homology_1 \oplus \Homology^1$ decomposes into two invariant subspaces~$V$ and~$W$.
	Both,~$V$ and~$W$ have dimension eight.
	The subspace~$V$ is symplectic and gives rise to a set of Pauli-operators acting on four qubits.
	The subspace~$W$, on the other hand, is maximally isotropic ($W^\perp = W$), in other words, it is a Lagrangian subspace.
	A basis of both sub-spaces can be found in \Cref{sec:invariant_subspaces}.
	The group~$G_{\tau_0}$ acts faithfully as~$\operatorname{Sp}_8(\mathbb{F}_2)$ on~$V$.
	Therefore, we obtain the full Clifford group~$\mathcal{C}_4$ when we restrict the logical subspace to the one induced by~$V$.

	\section{Conclusion}

	We have introduced a method to implement Clifford gates in CSS quantum codes using symmetries that we call $ZX$-dualities.
	A $ZX$-duality defines a fold-transversal operator, which is a quantum circuit with gates supported on the orbits of~$\tau$.
	In particular, we defined Hadamard-type gates~$H_\tau$ and phase-type gates $S_\tau$, derived from suitable $ZX$-dualities~$\tau$.
	For the surface code with~$\tau$ being a literal fold (see \Cref{fig:fold-surface}), this is equivalent to the construction of Moussa~\cite{moussa2016transversal}.
	Further, we explicitly discussed Bring's code as an example, showing that it is possible to generate the Clifford group~$\mathcal{C}_4$ by restricting the logical subspace.
	
	It would be interesting to understand the set of fold-transversal gates for a family of LDPC quantum codes, such as hyperbolic codes or those discussed in \Cref{sec:further_examples}.
	While one can verify the construction using a computer algebra system, as we have done here, this has its limitations.
	This is because the groups generated by the fold-transversal and permutation gates become very large, even for a moderate number of encoded qubits~$k$.

We leave to future work establishing the performance of these
fold-transversal gates under simulated noise and error correction rounds,
as well as their experimental realization.
	
	Finally, we expect our results to generalize to 
other code families that support a notion of Fourier duality:
non-CSS qubit stabilizer codes, qudit stabilizer codes, and subsystem codes.

\section*{Acknowledgements}	
N.P.B. acknowledges support through the EPSRC Prosperity Partnership in Quantum Software for Simulation and Modelling (Grant No. EP/S005021/1).
S.B. acknowledges financial support by the Foundation for
Polish Science through TEAM-NET project (contract no.
POIR.04.04.00-00-17C1/18-00).

\bibliographystyle{quantum}
\bibliography{bibliography}

\onecolumn\newpage
\appendix

\section{Matrices}

Here we give further details on associated matrices coming
from the construction of Bring's code in \Cref{sec:brings_code}.

\subsection{Qubit permutation}\label{sec:permutation_rep}
In the basis $\mathfrak{B}^{'}_{\Homology_1}$ we have
\begin{align*}
\tilde{\sigma}_{\Homology_1}(a) = 
\left(\begin{array}{rrrrrrrr}
1 & 0 & 0 & 0 & 0 & 0 & 0 & 0 \\
0 & 0 & 0 & 0 & 1 & 1 & 1 & 0 \\
0 & 1 & 1 & 0 & 0 & 0 & 0 & 1 \\
0 & 0 & 0 & 1 & 0 & 1 & 1 & 1 \\
0 & 0 & 0 & 0 & 1 & 0 & 0 & 0 \\
1 & 1 & 0 & 0 & 0 & 0 & 1 & 0 \\
1 & 0 & 0 & 0 & 1 & 0 & 1 & 0 \\
0 & 1 & 0 & 0 & 1 & 1 & 1 & 1
\end{array}\right) 
\qquad
\tilde{\sigma}_{\Homology_1}(b) = 
\left(\begin{array}{rrrrrrrr}
0 & 1 & 1 & 0 & 0 & 0 & 0 & 1 \\
1 & 0 & 0 & 0 & 0 & 1 & 1 & 1 \\
1 & 0 & 0 & 1 & 0 & 1 & 0 & 0 \\
0 & 0 & 1 & 0 & 1 & 0 & 0 & 1 \\
0 & 0 & 0 & 1 & 0 & 1 & 1 & 1 \\
0 & 1 & 1 & 0 & 1 & 0 & 0 & 0 \\
1 & 1 & 0 & 1 & 1 & 0 & 1 & 0 \\
1 & 0 & 0 & 1 & 0 & 0 & 1 & 1
\end{array}\right) 
\end{align*}
\begin{align*}
\tilde{\sigma}_{\Homology_1}(c) = 
\left(\begin{array}{rrrrrrrr}
1 & 0 & 0 & 0 & 0 & 1 & 1 & 1 \\
0 & 0 & 1 & 0 & 0 & 1 & 1 & 0 \\
0 & 0 & 1 & 0 & 0 & 0 & 0 & 0 \\
0 & 0 & 0 & 1 & 0 & 0 & 0 & 0 \\
0 & 1 & 0 & 0 & 1 & 0 & 0 & 1 \\
0 & 1 & 0 & 1 & 0 & 0 & 1 & 0 \\
0 & 0 & 1 & 1 & 0 & 0 & 1 & 0 \\
0 & 1 & 1 & 0 & 0 & 1 & 1 & 1
\end{array}\right)
\end{align*}
and in the basis $\mathfrak{B}^{'}_{\Homology^1}$ we have $\tilde{\sigma}_{\Homology^1}(x) = \tilde{\sigma}_{\Homology_1}(x)^\top$, for $x\in\{a,b,c\}$.

Note that the matrices operate on \emph{row vectors} from the \emph{right}.

\subsection{$G_{\tau_0}$-Invariant subspaces of $\Homology_1 \oplus \Homology^1$}\label{sec:invariant_subspaces}
A symplectic basis of the symplectic space $V\subset \Homology_1 \oplus \Homology^1$ in the basis $\mathfrak{B}^{'}_{\Homology_1 \oplus \Homology^1}$: 
\begin{align*}
	\left(\begin{array}{rrrrrrrr|rrrrrrrr}
	0 & 1 & 0 & 0 & 1 & 0 & 0 & 0 & 0 & 0 & 0 & 0 & 0 & 0 & 0 & 0 \\
	0 & 1 & 1 & 0 & 0 & 0 & 0 & 0 & 0 & 0 & 0 & 0 & 0 & 0 & 0 & 0 \\
	1 & 1 & 1 & 1 & 1 & 0 & 1 & 0 & 0 & 0 & 0 & 0 & 0 & 0 & 0 & 0 \\
	1 & 0 & 0 & 1 & 0 & 1 & 1 & 1 & 0 & 0 & 0 & 0 & 0 & 0 & 0 & 0 \\
	\hline
	0 & 0 & 0 & 0 & 0 & 0 & 0 & 0 & 0 & 1 & 1 & 0 & 0 & 0 & 0 & 0 \\
	0 & 0 & 0 & 0 & 0 & 0 & 0 & 0 & 0 & 1 & 0 & 0 & 1 & 0 & 0 & 0 \\
	0 & 0 & 0 & 0 & 0 & 0 & 0 & 0 & 1 & 0 & 0 & 1 & 0 & 1 & 1 & 0 \\
	0 & 0 & 0 & 0 & 0 & 0 & 0 & 0 & 0 & 1 & 1 & 0 & 1 & 1 & 1 & 1
	\end{array}\right)
\end{align*}
A basis of the Lagrangian space $W\subset \Homology_1 \oplus \Homology^1$ in the basis $\mathfrak{B}^{'}_{\Homology_1 \oplus \Homology^1}$: 
\begin{align*}
	\left(\begin{array}{rrrrrrrr|rrrrrrrr}
	1 & 0 & 0 & 0 & 1 & 1 & 0 & 1 & 0 & 0 & 0 & 0 & 0 & 0 & 0 & 0 \\
	0 & 0 & 0 & 0 & 0 & 1 & 0 & 1 & 0 & 0 & 0 & 0 & 0 & 0 & 0 & 0 \\
	0 & 0 & 1 & 1 & 0 & 0 & 0 & 0 & 0 & 0 & 0 & 0 & 0 & 0 & 0 & 0 \\
	0 & 1 & 0 & 0 & 0 & 0 & 1 & 1 & 0 & 0 & 0 & 0 & 0 & 0 & 0 & 0 \\
	\hline
	0 & 0 & 0 & 0 & 0 & 0 & 0 & 0 & 1 & 0 & 0 & 0 & 1 & 1 & 1 & 1 \\
	0 & 0 & 0 & 0 & 0 & 0 & 0 & 0 & 0 & 0 & 0 & 0 & 0 & 1 & 1 & 1 \\
	0 & 0 & 0 & 0 & 0 & 0 & 0 & 0 & 0 & 0 & 1 & 1 & 0 & 0 & 0 & 0 \\
	0 & 0 & 0 & 0 & 0 & 0 & 0 & 0 & 0 & 1 & 0 & 0 & 0 & 0 & 1 & 0
	\end{array}\right)
\end{align*}

In the basis of V above, the action of the generators of $G_{\tau_0}\cong \operatorname{Sp}_8(\mathbb{F}_2)$ restricted to $V$ is given by the following matrices:
\begin{align*}
	\tilde{\sigma}_r |_V = 
	\left(\begin{array}{rrrr|rrrr}
	1 & 0 & 0 & 1 & 0 & 0 & 0 & 0 \\
	1 & 0 & 1 & 0 & 0 & 0 & 0 & 0 \\
	1 & 0 & 0 & 0 & 0 & 0 & 0 & 0 \\
	1 & 1 & 0 & 0 & 0 & 0 & 0 & 0 \\
	\hline
	0 & 0 & 0 & 0 & 1 & 0 & 0 & 1 \\
	0 & 0 & 0 & 0 & 1 & 0 & 1 & 1 \\
	0 & 0 & 0 & 0 & 0 & 1 & 0 & 0 \\
	0 & 0 & 0 & 0 & 1 & 1 & 0 & 0
	\end{array}\right) 
	\qquad
	\tilde{\sigma}_s |_V = 
	\left(\begin{array}{rrrr|rrrr}
	0 & 0 & 0 & 1 & 0 & 0 & 0 & 0 \\
	0 & 0 & 1 & 1 & 0 & 0 & 0 & 0 \\
	1 & 0 & 0 & 1 & 0 & 0 & 0 & 0 \\
	0 & 1 & 0 & 1 & 0 & 0 & 0 & 0 \\
	\hline
	0 & 0 & 0 & 0 & 0 & 0 & 0 & 1 \\
	0 & 0 & 0 & 0 & 0 & 0 & 1 & 0 \\
	0 & 0 & 0 & 0 & 1 & 1 & 0 & 0 \\
	0 & 0 & 0 & 0 & 0 & 1 & 0 & 1
	\end{array}\right) 
\end{align*}
\begin{align*}
	\tilde{H}_{\tau_0} |_V = 
	\left(\begin{array}{rrrr|rrrr}
	0 & 0 & 0 & 0 & 1 & 0 & 0 & 0 \\
	0 & 0 & 0 & 0 & 0 & 1 & 0 & 0 \\
	0 & 0 & 0 & 0 & 0 & 0 & 0 & 1 \\
	0 & 0 & 0 & 0 & 0 & 0 & 1 & 1 \\
	\hline
	1 & 0 & 0 & 0 & 0 & 0 & 0 & 0 \\
	0 & 1 & 0 & 0 & 0 & 0 & 0 & 0 \\
	0 & 0 & 1 & 1 & 0 & 0 & 0 & 0 \\
	0 & 0 & 1 & 0 & 0 & 0 & 0 & 0
	\end{array}\right) 
	\qquad
	\tilde{S}_{\tau_0} |_V = 
	\left(\begin{array}{rrrr|rrrr}
	1 & 0 & 0 & 0 & 0 & 0 & 0 & 0 \\
	0 & 1 & 0 & 0 & 0 & 0 & 0 & 0 \\
	0 & 0 & 1 & 0 & 0 & 0 & 0 & 0 \\
	0 & 0 & 0 & 1 & 0 & 0 & 0 & 0 \\
	\hline
	1 & 0 & 0 & 0 & 1 & 0 & 0 & 0 \\
	0 & 1 & 0 & 0 & 0 & 1 & 0 & 0 \\
	0 & 0 & 1 & 1 & 0 & 0 & 1 & 0 \\
	0 & 0 & 1 & 0 & 0 & 0 & 0 & 1
	\end{array}\right) 
\end{align*}
Note that matrices operate on \emph{row vectors} from the \emph{right}.

\section{Further Examples}\label{sec:further_examples}
Here we discus some further examples of fold-transversal gates and the groups they generate.
We do not go through all the details, as we did for Bring's code.
However, the derivation of these results follows along the same steps.

\subsection{Hyperbolic Surface Codes}

We list some more examples of $\{5,5\}$ hyperbolic surface codes
obtained from the construction in \Cref{sec:surface-codes},
the first of which is Bring's code.
The table shows the parameters $[[n,k,d]]$ for the code $C$,
and the order of the automorphism group $\Aut(C)$ which is
equal to the number of $ZX$-dualities.

\begin{center}
\begin{tabular}{| c |c | }
\hline
 Code parameters & $|\mathcal{D}_{ZX}| = |\Aut(C)|$ \\ 
\hline\hline
 $[[30,8,3]]$   & 120 \\  
 $[[40,10,4]]$  & 160 \\
 $[[80,18,5]]$  & 320 \\
 $[[150,32,6]]$ & 600 \\
 $[[180,38,4]]$ & 720 \\
 $[[330,68,6]]$ & 1320 \\
 $[[480,95,?]]$ & 1920 \\
\hline
\end{tabular}
\end{center}

\subsection{Hypergraph Product Code}
Consider the $[6,2,4]$ code defined by the parity check matrix
\begin{align*}
	\left(\begin{array}{rrrr}
		1 & 1 & 0 & 0 \\
		0 & 1 & 0 & 0 \\
		1 & 0 & 1 & 0 \\
		0 & 0 & 1 & 0 \\
		1 & 0 & 0 & 1 \\
		0 & 0 & 0 & 1
	\end{array}\right) 
\end{align*}
acting on row vectors.
The hypergraph product of this code with its dual gives a $[[52,4,4]]$ quantum code with checks of weights 3, 4 and 5.
The automorphism group of the code, as defined in \Cref{sec:ZX-dualities}, is $(\operatorname{S}_3 \times \operatorname{S}_3) \rtimes \operatorname{C}_2$.
This can be seen from the fact that $\operatorname{S}_3$ is the symmetry group of the input code and we have the exchange of $X$- and $Z$-checks, giving the semi-direct product with~$\operatorname{C}_2$.
There are four $ZX$-dualities and all of them fulfill the requirements of the theorems in \Cref{sec:fold_transversal_gates}.
Using a suitable choice of basis for $\Homology_1\oplus \Homology^1$, we can construct eight Hadamard- and phase-type gates.
Together with the permutation gates they generate the group $(\operatorname{A}_6 \times \operatorname{A}_6) \rtimes \operatorname{D}_4$.
This is an index-45696 subgroup of~$\operatorname{Sp}_8(\mathbb{F}_2)$.

As for Bring's code, we can obtain the full group symplectic group~$\operatorname{Sp}_8(\mathbb{F}_2)$ by adding a single entangling gate, e.g. by using surgery with an ancillary surface code~\cite{cohen2021low}.
Alternatively, the largest symplectic subgroup of $(\operatorname{A}_6 \times \operatorname{A}_6) \rtimes \operatorname{D}_4$ is $\operatorname{Sp}_4(\mathbb{F}_2)$, so we can obtain the full Clifford group on two logical qubits when restricting the code space accordingly.

\subsection{Balanced Product Code}
Balanced product quantum codes are a generalization of the hypergraph product codes, allowing for the construction of quantum codes that can achieve higher encoding rates and larger distances~\cite{breuckmann2021balanced}.
The balanced product between two classical codes, that share a common symmetry, can be understood as the hypergraph product between the two, with the symmetry subsequently factored out.
We refer to \cite{PRXQuantum.2.040101,breuckmann2021balanced} for more details.
Symmetry groups play a prominent role in the construction of balanced product codes, so that they are amenable to our scheme.

For symmetric classical codes, let us consider the extended binary Golay code $\mathcal{G}_{24}$ with parameters $[24,12,8]$ and its dual $\mathcal{G}_{24}^\top$.
It turns out that, coincidentally, the parity checks of the Golay code can also be derived from the great dodecahedron (\Cref{fig:great_dodecahedron}), see \cite{golay_code} for details.
The automorphism group of this parity check matrix of~$\mathcal{G}_{24}$ is thus the same as that of Bring's code~$\operatorname{S}_5 \times \operatorname{C}_2$ (see \Cref{sec:construction_bring}).
The alternating group~$\operatorname{A}_5$ is a unique subgroup.
The balanced product code $\mathcal{G}_{24} \otimes_{\operatorname{A}_5} \mathcal{G}_{24}^\top$ is a $[[20,4,3]]$ code with check weights 4, 5 and 7.
It has 10 $ZX$-dualities, giving rise to Hadamard-type gates.
However, only one of the $ZX$-dualities satisfies the requirements of \Cref{thm:S_type_gate}, giving rise to a single phase-type gate.
The group of permissible qubit permutations is $\operatorname{C}_2\times \operatorname{S}_4$ and it has a faithful representation on $\Homology_1 \oplus \Homology^1$.
The permutation gates, the Hadamard-type gates and the phase-type gate generate the group $\operatorname{S}_3\times \operatorname{S}_4$, which is an index-329011200 subgroup of $\operatorname{Sp}_8(\mathbb{F}_4)$.
Adding a single entangling gate generates the full symplectic group~$\operatorname{Sp}_8(\mathbb{F}_2)$.

Alternatively, consider the smaller subgroup $H = \operatorname{C}_2 \times (\operatorname{C}_5 \rtimes \operatorname{C}_4)$.
The balanced product code $\mathcal{G}_{24} \otimes_H \mathcal{G}_{24}^\top$ is a $[[30,6,3]]$ code.
Similarly, the permutation gates, the Hadamard-type gates and the phase-type gate generate an index-135491300609753088000 subgroup of $\operatorname{Sp}_{12}(\mathbb{F}_4)$.
Adding a single entangling gate generates the full symplectic group~$\operatorname{Sp}_{12}(\mathbb{F}_2)$.

A $[[60,12,5]]$ code can be obtained by taking the balanced product over the even smaller subgroup $\operatorname{C}_3 \rtimes \operatorname{C}_4$.
Unfortunately, the space $\Homology_1 \oplus \Homology^1$ is 24-dimensional and thus too large for our computational methods.

\subsection{Block Code}
The following example shows that the fold-transversal gates do not always suffice to generate large groups.
Consider the $[[16,4,4]]$ code, obtained by concatenating the $[[4,2,2]]$ code with itself.
It is a weakly self-dual CSS code with classical code defined by the check matrix
\begin{align*}
	\left(\begin{array}{rrrrrr}
	1 & 1 & 1 & 0 & 0 & 0 \\
	1 & 1 & 0 & 1 & 0 & 0 \\
	1 & 1 & 0 & 0 & 1 & 0 \\
	1 & 1 & 0 & 0 & 0 & 1 \\
	1 & 0 & 1 & 0 & 0 & 0 \\
	1 & 0 & 0 & 1 & 0 & 0 \\
	1 & 0 & 0 & 0 & 1 & 0 \\
	1 & 0 & 0 & 0 & 0 & 1 \\
	0 & 0 & 1 & 0 & 0 & 0 \\
	0 & 0 & 0 & 1 & 0 & 0 \\
	0 & 0 & 0 & 0 & 1 & 0 \\
	0 & 0 & 0 & 0 & 0 & 1 \\
	0 & 1 & 1 & 0 & 0 & 0 \\
	0 & 1 & 0 & 1 & 0 & 0 \\
	0 & 1 & 0 & 0 & 1 & 0 \\
	0 & 1 & 0 & 0 & 0 & 1
	\end{array}\right)
\end{align*}
A natural basis of $H_1$ and $H^1$ is 
\begin{align*}
	\left(\begin{array}{rrrrrrrrrrrrrrrr}
	1 & 1 & 0 & 0 & 1 & 1 & 0 & 0 & 0 & 0 & 0 & 0 & 0 & 0 & 0 & 0 \\
	1 & 1 & 0 & 0 & 0 & 0 & 0 & 0 & 0 & 0 & 0 & 0 & 1 & 1 & 0 & 0 \\
	0 & 1 & 1 & 0 & 0 & 1 & 1 & 0 & 0 & 0 & 0 & 0 & 0 & 0 & 0 & 0 \\
	0 & 1 & 1 & 0 & 0 & 0 & 0 & 0 & 0 & 0 & 0 & 0 & 0 & 1 & 1 & 0
	\end{array}\right)
\end{align*}
There exist 20 $ZX$-dualities, giving rise to 8 distinct Hadamard-type gates on $\Homology_1 \oplus \Homology^1$.
Only one of the $ZX$-dualities satisfies the requirements of \Cref{thm:S_type_gate}, giving rise to a phase-type gate.
The group of permissible qubit permutations is $\operatorname{C}_2\times \operatorname{S}_4$.
However, unlike in the previous examples, the induced representation on $\Homology_1 \oplus \Homology^1$ is not faithful, as its image is isomorphic to the dihedral group~$\operatorname{D}_6$.
The permutation gates, the Hadamard-type gates and the phase-type gate generate the group $\operatorname{C}_2\times \operatorname{S}_3\times \operatorname{S}_3$.
Unlike in the previous examples, note even adding all $\tilde{\operatorname{CZ}}_{i,j}$-gates does not give rise to the full symplectic group~$\operatorname{Sp}_8(\mathbb{F}_2)$, but to the orthogonal group~$\operatorname{O}^+_8(\mathbb{F}_2)$, which is an index-272 subgroup.

\end{document}